\documentclass[12pt]{article}

\usepackage{amsmath,amsthm,amssymb,mathrsfs}
\usepackage{graphicx,mdframed}

\font\tencmmib=cmmib10 \skewchar\tencmmib '60
\newfam\cmmibfam
\textfont\cmmibfam=\tencmmib

\def\lessim{\ \lower4pt\hbox{$
		\buildrel{\displaystyle <}\over\sim$}\ }
\def\gessim{\ \lower4pt\hbox{$\buildrel{\displaystyle >}
		\over\sim$}\ }

\newcommand{\e}{\mathbb{E}}
\newcommand{\p}{\mathbb{P}}

\newcommand{\pr}{\mathrm{Pr}}

\newtheorem{lemma}{\bf Lemma}

\newtheorem{theorem}{\bf Theorem}

\newtheorem{remark}{\bf Remark}

\newtheorem{proposition}{\bf Proposition}

\newmdtheoremenv{theo}{Theorem}

\newenvironment{Proof of lemma}{\noindent{\bf Proof of Lemma}}{\hfill$\Box$\newline}
\newenvironment{Proof of theorem}{\noindent{\it Proof of Theorem}}{\hfill\scriptsize{$\Box$}\newline}
\newenvironment{Proof of theorems}{\noindent{\bf Proof of Theorems}}{\hfill$\Box$\newline}
\newenvironment{Proof of proposition}{\noindent{\bf Proof of Proposition}}{\hfill$\Box$\newline}
\newenvironment{Proof of propositions}{\noindent{\bf Proof of Propositions}}{\hfill$\Box$\newline}
\newenvironment{Proof of exercise}{\noindent{\it Proof of Exercise:}}{\hfill$\Box$}
\begin{document}
	
			\title{On the Almeida-Thouless transition line \\
				in the Sherrington-Kirkpatrick model\\
				 with centered Gaussian external field}
	
	%\author{Wei-Kuo Chen, Wai-Kit Lam}\thanks{School of Mathematics, University of Minnesota. Email: wkchen@umn.edu. W.-K. Chen's research is partially supported by NSF grant DMS-17-52184.}
	\author{Wei-Kuo Chen\thanks{University of Minnesota. Email: wkchen@umn.edu.  Partly supported by NSF grant DMS-17-52184}}
	
	\maketitle
	
	\begin{abstract}	
		We study the phase transition of the free energy in the Sherrington-Kirkpatrick mean-field spin glass model with centered Gaussian external field. We show that the corresponding Almeida-Thouless line is the correct transition curve that distinguishes between the replica symmetric and replica symmetry breaking solutions in the Parisi formula.
	\end{abstract}
	
	\section{Introduction and main results}

The famous Sherrington-Kirkpatrick (SK) mean-field spin glass model was introduced in \cite{SK72} aiming to explain some unusual magnetic behavior of certain alloys. By means of the replica method, it was also proposed in \cite{SK72} that the thermodynamic limit of the free energy in the SK model  can be solved  by the replica symmetric ansatz at very high temperature.  A complete picture was later settled in the seminal works \cite{Parisi79,Parisi80} of Parisi, in which he adapted an ultrametric ansatz and deduced a variational formula for the limiting free energy at all temperature, known as the Parisi formula. This formula was rigorously established by Talagrand \cite{Tal06} utilizing the replica symmetry breaking bound discovered by Guerra \cite{Guerra03}. See \cite{MPV87} for physicists' studies of the SK model as well as \cite{Pan13,Tal111,Tal112} for the recent mathematical progress.

For any $N\geq 1,$ the Hamiltonian of the SK model is defined as 
\begin{align}\label{SK}
	-H_N(\sigma)=\frac{\beta}{\sqrt{N}}\sum_{1\leq i<j\leq N}g_{ij}\sigma_i\sigma_j+h\sum_{i=1}^N\sigma_i,\,\,\forall\sigma\in \{-1,1\}^N,
\end{align}
where $g_{ij}$'s are i.i.d. standard Gaussian. The parameters $\beta>0$ and $h\in \mathbb{R}$ are the (inverse) temperature and external field, respectively. Define the free energy as $$F_N(\beta,h)=\frac{1}{N}\log\sum_{\sigma\in \{-1,1\}^N}e^{H_N(\sigma)}.$$
The famous Parisi formula \cite{Pan13,Tal112} asserts that almost surely,
\begin{align}\label{eq2}
	\lim_{N\to\infty}F_N(\beta,h)=\min_{\alpha\in \pr([0,1])}{P}(\alpha).
\end{align}
Here, $\pr([0,1])$ is the collection of all probability distribution functions on $[0,1]$ equipped with the $L^1(dx)$ distance and $P$ is a functional on  $\pr([0,1])$ defined as
\begin{align*}
	{P}(\alpha)&=\log 2+\e\Phi_{\alpha}(0,h)-\frac{\beta^2}{2}\int_0^1s\alpha(s)ds,
\end{align*}
where  $\Phi_\alpha$ is the weak solution \cite{JT16} to
\begin{align}\label{eq1}
	\partial_s\Phi_\alpha(s,x)&=-\frac{\beta^2}{2}\bigl(\partial_{xx}\Phi_\alpha(s,x)+\alpha(s)(\partial_x\Phi_\alpha(s,x))^2\bigr),\,\,(s,x)\in [0,1]\times \mathbb{R}
\end{align}
with boundary condition $\Phi(1,x)=\log\cosh x.$ It is known \cite{AC15} that the Parisi formula has a unique minimizer denoted by $\alpha_P.$ We say that the Parisi formula is solved by the replica symmetric solution if $\alpha_P=1_{[q,1]}$ for some $q\in [0,1]$ and is solved by the replica symmetry breaking solution if otherwise.

For any $\beta,h>0,$ let $q=q(\beta,h)$ be the unique solution (see \cite{Guerra01} and \cite[Proposition A.14.1]{Tal111}) to 
\begin{align}
	\label{fxpt}
	q=\e\tanh^2(h+\beta z\sqrt{q}),
\end{align}
where $z$ is standard Gaussian. In \cite{AT78}, it was conjectured that for $\beta,h>0$, the SK model is solved by the replica symmetric solution if and only if $(\beta,h)$ satisfies
$$
\beta^2\e\frac{1}{\cosh^4(h+\beta z\sqrt{q})}\leq 1.
$$ 
In other words, for $\beta,h>0,$ the following equation, known as the Almeida-Thouless line,
\begin{align}\label{atline0}
	\beta^2\e\frac{1}{\cosh^4(h+\beta z\sqrt{q})}=1,
\end{align}
characterizes the transition between the replica symmetric and replica symmetry breaking solutions.
Toninelli \cite{Toninelli02} proved that above the AT line, i.e, $\beta^2\e {\cosh^{-4}(h+\beta z\sqrt{q})}>1$, the solution to the Parisi formula is replica symmetry breaking. Later Talagrand \cite{Tal112} and Jagannath-Tobasco \cite{JT17} showed that inside the AT line, $\beta^2\e {\cosh^{-4}(h+\beta z\sqrt{q})}\leq1$, there exist fairly large regimes in which the Parisi formula is solved by replica symmetric solution. As parts of their regimes are not up to the AT line,  verifying the exactness of the AT line remains open.

In this short note, we investigate the SK model with centered Gaussian external field and show that the corresponding AT line is indeed the transition line distinguishing between the replica symmetric and replica symmetry breaking solutions in the Parisi formula. For $\beta,h>0$, the Hamiltonian of the SK model with centered Gaussian external field is defined as
\begin{align*}
	-	\mathcal{H}_N(\sigma)&=\frac{\beta}{\sqrt{N}}\sum_{1\leq i<j\leq N}g_{ij}\sigma_i\sigma_j+h\sum_{i=1}^N\xi_i\sigma_i,\,\,\forall \sigma\in \{-1,1\}^N,
\end{align*}
where $(g_{ij})_{1\leq i<j\leq N}$ and $(\xi_i)_{1\leq i\leq N}$ are i.i.d. standard normal and are independent of each other. The free energy associated to this Hamiltonian is given by
\begin{align*}
	\mathcal{F}_N(\beta,h)=\frac{1}{N}\e\log \sum_{\sigma\in \{-1,1\}^N}\exp \bigl(-\mathcal{H}_N(\sigma)\bigr).
\end{align*}
In a similar manner, the limiting free energy can be expressed by the Parisi formula as in \eqref{eq1} with a replacement of $h$ by $h\xi$ for $\xi$ a standard Gaussian random variable, more precisely, we have that almost surely,
\begin{align*}
	\lim_{N\to\infty}\mathcal{F}_N(\beta,h)=\min_{\alpha\in \pr([0,1])}\mathcal{P}(\alpha),
\end{align*}
where 
\begin{align*}
	\mathcal{P}(\alpha)&=\log 2+\e\Phi_{\alpha}(0,h\xi)-\frac{\beta^2}{2}\int_0^1s\alpha(s)ds,
\end{align*}
and $\Phi_\alpha$ is defined as \eqref{eq1}. Note that the proof of this formula is identically the same as those in \cite{Pan13,Tal112} with no essential modifications.  Moreover, as  in \cite{AC15}, the Parisi formula here also has a unique minimizer denoted by $\alpha_P.$ We again say that the Parisi formula is solved by the replica symmetric solution if $\alpha_P=1_{[q,1]}$ for some $q\in [0,1]$ and is solved by the replica symmetry breaking solution if otherwise.

\begin{figure}\label{figure1}
	\includegraphics[width=5cm]{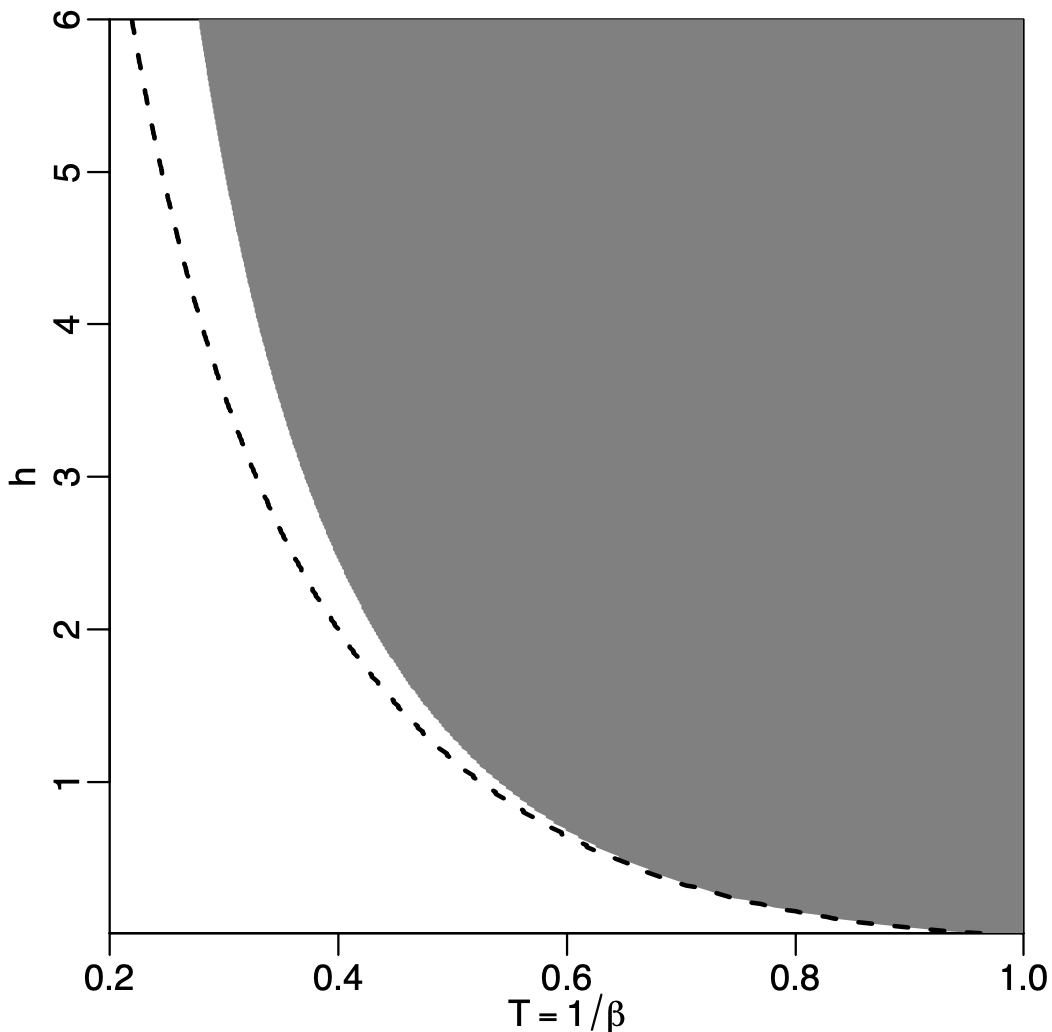}
	\centering
	\caption{This figure describes the phase transition in the SK model with centered Gaussian external field. The gray area is the regime of replica symmetric solutions, while the white area is the regime of replica symmetry breaking solutions. The boundary between these two regimes is the AT line \eqref{atline} corresponding to our model. The black-dash line is the original SK model with deterministic external field \eqref{atline0}.}
\end{figure}

The AT line corresponds to our model is formulated analogously. Let $z,\xi$ be i.i.d. standard Gaussian random variables. For $\beta,h>0,$ let $q=q(\beta,h)$ be the unique fixed point of 
\begin{align}\label{fixedpt}
	q&=\e \tanh^2(h\xi+\beta z\sqrt{q}),
\end{align} 
where the existence and uniqueness of $q$ are guaranteed thanks to the Latala-Guerra lemma.\footnote{
	In \cite{Guerra01} and \cite[Proposition A.14.1]{Tal111}, it is known that for any $r\in \mathbb{R},$ $x^{-1} {\e \tanh^2(r+\beta z\sqrt{x})}$ is a strictly decreasing function in $x>0$, which implies that $f(x):=x^{-1} {\e \tanh^2(h\xi+\beta z\sqrt{x})}$ is also strictly decreasing on $(0,\infty)$. Since $f(\infty)=0$ and $f(0)=\infty$, there exists a unique $x_0$ such that $f(x_0)=1.$ This ensures that \eqref{fixedpt} has a unique solution.
} The AT line associated to the SK model with centered Gaussian external field is the collection of all $(\beta,h)\in (0,\infty)\times (0,\infty)$ satisfying 
\begin{align}\label{atline}
	\beta^2\e \frac{1}{\cosh^4(h\xi+\beta z\sqrt{q})}=1.
\end{align}
The following is our main result.

\begin{theorem}\label{thm1}
	Consider the SK model with centered Gaussian external field. For any $\beta>0$ and $h>0,$ the Parisi formula exhibits the replica symmetric solution if and only if $(\beta,h)$ lies inside the AT line, i.e., \begin{align}\label{inat}
		\beta^2\e \frac{1}{\cosh^4(h\xi+\beta z\sqrt{q})}\leq 1.
	\end{align}
\end{theorem} 

The rest of the paper is organized as follows. In Section \ref{sec2}, we gather some results regarding the directional derivative of the functional $\mathcal{P}$ and some consequences following the first order optimality of the Parisi variational formula. The proof of Theorem \ref{thm1} is presented in Section \ref{sec3}.

\smallskip

{\noindent \bf Acknowledgements.} 	The author thanks the anonymous referees for some helpful comments regarding the presentation of this paper. Special thanks are due to Si Tang for running the simulations in Figures \ref{figure1} and \ref{figure3}.

\section{Preliminary results}\label{sec2}

Let $\xi$ be standard Gaussian. Let $\alpha\in \mbox{Pr}[0,1]$. Conditionally on $\xi$, let $(X_{\alpha}^\xi(s))_{0\leq s\leq 1}$ be the solution to the following stochastic differential equation with boundary condition $X_{\alpha}^\xi(0)=h\xi$,
\begin{align*}
	dX_{\alpha}^\xi(s)&=\beta^2\alpha(s)\partial_x\Phi_{\alpha}(s,X_{\alpha}^\xi(s))ds+\beta dW(s),
\end{align*}
where $(W(s))_{s\in [0,1]}$ is a standard Brownian motion independent of $\xi.$ Set
\begin{align*}
	\phi_\alpha(s)=\e \bigl(\partial_x\Phi_{\alpha}(s,X_{\alpha}^\xi(s))\bigr)^2,\,\,0\leq s\leq 1.
\end{align*}
The following proposition establishes the directional derivative of $\mathcal{P}$ in terms of $\phi_\alpha.$

\begin{proposition}\label{prop1}
	For $\alpha_0,\alpha_1\in \mbox{Pr}[0,1]$ and $\theta\in [0,1],$ set $\alpha_\theta=(1-\theta)\alpha_0+\theta\alpha_1.$ The right-derivative of $\theta\mapsto\mathcal{P}(\alpha_\theta)$ at zero is given by
	\begin{align*}
		\frac{d\mathcal{P}(\alpha_\theta)}{d\theta^+}\Bigr|_{\theta=0}=\frac{\beta^2}{2}\int_0^1\bigl(\alpha_1(s)-\alpha_0(s)\bigr)\bigl(s-\phi_{\alpha_0}(s)\bigr)ds.
	\end{align*}
	Furthermore, $\alpha_0$ is the minimizer of $\min_{\alpha\in \mbox{\footnotesize Pr}[0,1]}\mathcal{P}(\alpha)$ if and only if $\frac{d\mathcal{P}(\alpha_\theta)}{d\theta^+}\Bigr|_{\theta=0}\geq 0$ for all $\alpha_1\in \mbox{Pr}[0,1].$
\end{proposition}

\begin{proposition}\label{prop2}
	If $\alpha_0$ is the minimizer of $\mathcal{P}$, then every point in the support of $\alpha_0$ must satisfy $\phi_{\alpha_0}(s)=s.$
\end{proposition}

In the case of the SK model with deterministic external field \eqref{SK}, the two propositions above were established in Theorem 2 and Proposition 1 in \cite{Chen17}, respectively, where the statements are essentially the same as Propositions \ref{prop1} and \ref{prop2} with a replacement of $h\xi$ by $h.$ As the proofs in \cite{Chen17} can be directly applied to Propositions \ref{prop1} and \ref{prop2} with no major changes, we do not reproduce the details here. The following lemma provides a simpler expression for $\phi_{\alpha}$.

\begin{lemma}\label{lem1}
	Let $\alpha\in \mbox{Pr}[0,1].$ For any $s\in [0,1]$, we have that 
	\begin{align*}
		\phi_\alpha(s)&=\e \partial_{x}\Phi_{\alpha}(s,M(s))^2\exp\Bigl(\int_0^s (\Phi_{\alpha}(s,M(s))-\Phi_{\alpha}(t,M(t)))\alpha(dt)\Bigr),
	\end{align*}
	where $M(s):=h\xi+\beta W(s).$
\end{lemma}

\begin{proof}
	Conditionally on $\xi$, let ${\e}^\xi$ be the expectation associated to the probability measure 
	\begin{align*}
		d{\p}^\xi=	\exp\Bigl(-\int_0^1\beta\alpha(r)\partial_{x}\Phi(r,X_\alpha^\xi(r))dW(r)-\frac{1}{2}\int_0^1\beta^2 \alpha(r)^2\partial_x\Phi_{\alpha}(r,X_\alpha^\xi(r))^2dr\Bigr)d\p.
	\end{align*}
	From the Girsanov theorem, under ${\e}^\xi$, $$
	\Bigl(\int_0^s\beta \alpha(r)\partial_x\Phi_{\alpha}(r,X_\alpha^\xi(r))dr+W(s)\Bigr)_{0\leq s\leq 1}
	$$ is a standard Brownian motion for which we denote by $(W^\xi(s))_{0\leq s\leq 1}.$ Set $M^\xi(s)=h\xi+\beta W^\xi(s).$ Under $\e^\xi,$ write
	\begin{align*}
		&\int_0^1\beta\alpha(r)\partial_{x}\Phi(r,X_\alpha^\xi(r))dW(r)+\frac{1}{2}\int_0^1\beta^2 \alpha(r)^2\partial_x\Phi_{\alpha}(r,X_\alpha^\xi(r))^2dr\\
		&=\int_0^1\beta\alpha(r)\partial_{x}\Phi(r,X_\alpha^\xi(r))\Bigl(\beta \alpha(r)\partial_x\Phi_\alpha(r,X_\alpha^\xi(r)) dr+dW(r)\Bigr)\\
		&\qquad\qquad-\frac{1}{2}\int_0^1\beta^2 \alpha(r)^2\partial_x\Phi_{\alpha}(r,X_\alpha^\xi(r))^2dr\\
		&=\int_0^1\beta\alpha(r)\partial_{x}\Phi(r, M^\xi(r))dW^\xi(r)-\frac{1}{2}\int_0^1\beta^2 \alpha(r)^2\partial_x\Phi_{\alpha}(r,  M^\xi(r))^2dr.
	\end{align*}
	Consequently, 
	\begin{align}
		\nonumber \e\bigl[\bigl(\partial_x\Phi_\alpha(s,X_{\alpha}^\xi(s))\bigr)^2\big|\xi\bigr]&=\e^\xi \bigl(\partial_x\Phi_\alpha(s,M^\xi(s))\bigr)^2\exp\Bigl(\int_0^1\beta\alpha(r)\partial_{x}\Phi_\alpha(r, M^\xi(r))dW^\xi(r)\\
		\label{eq3}&\qquad\qquad\qquad\qquad-\frac{1}{2}\int_0^1\beta^2 \alpha(r)^2\partial_x\Phi_{\alpha}(r, M^\xi(r))^2dr\Bigr).
	\end{align}
	Next, note that from It\^o's formula and \eqref{eq1}, for $0\leq t<s\leq 1,$
	\begin{align*}
		&\Phi_\alpha(s, M^\xi(s))-\Phi_\alpha(t, M^\xi(t))\\
		&=-\frac{\beta^2}{2}\int_t^s\alpha(r)(\partial_x\Phi_\alpha(r, M^\xi(r)))^2dr+\beta \int_t^s\partial_x\Phi_\alpha(r, M^\xi(r))W^\xi(dr).
	\end{align*} 
	and then from the Fubini theorem,
	\begin{align*}
		&\int_0^s\bigl(\Phi_\alpha(s, M^\xi(s))-\Phi_\alpha(t, M^\xi(t))\bigr)\alpha(dt)\\
		&=-\frac{\beta^2}{2}\int_0^s\alpha(r)^2(\partial_x\Phi_\alpha(r,M^\xi(r)))^2dr+\beta \int_0^s\alpha(r)\partial_x\Phi_\alpha(r, M^\xi(r))W^\xi(dr).
	\end{align*}
	Plugging this into \eqref{eq3} yields that
	\begin{align*}
		\e \bigl[\bigl(\partial_x\Phi_\alpha(s,X_\alpha^\xi(s))\bigr)^2\big|\xi\bigr]&=\e^\xi \Bigl[\partial_{x}\Phi_{\alpha}(s,M^\xi(s))^2\\
		&\qquad\exp\Bigl(\int_0^s \bigl(\Phi_{\alpha}(s,M^\xi(s))-\Phi_{\alpha}(t,M^\xi(t))\bigr)\alpha(dt)\Bigr)\Bigr]\\
		&=\e \Bigl[\partial_{x}\Phi_{\alpha}(s,M(s))^2\\
		&\qquad\exp\Bigl(\int_0^s \bigl(\Phi_{\alpha}(s,M(t))-\Phi_{\alpha}(t,M(t))\bigr)\alpha(dt)\Bigr)\Big|\xi\Bigr],
	\end{align*}
	where we used that $(M(s))_{0\leq s\leq 1}\stackrel{d}{=}(M^\xi(s))_{0\leq s\leq 1}$ conditionally on $\xi.$
	Finally, taking expectation in $\xi$ completes our proof.
\end{proof}

While Lemma \ref{lem1} holds for any $\alpha$, the important case that we shall use in our main proof is when $\alpha=1_{[q,1]}$ for some $q\in [0,1].$ In this case, one can compute $\phi_\alpha$ more explicitly. Let $z,z'$ be i.i.d. standard Gaussian independent of $\xi.$ Denote by $\e'$ the expectation with respect to $z'$ only. 

\begin{lemma}\label{lem2}
	Let $\alpha=1_{[q,1]}$ for some $q\in [0,1]$. We have that
	\begin{align}\label{lem2:eq1}
		\phi_\alpha(s)&=
		\left\{
		\begin{array}{ll}
			\e  \Bigl(\frac{\e'\tanh^2 (h\xi+\beta z\sqrt{q}+\beta z'\sqrt{s-q})\cosh (h\xi+\beta z\sqrt{q}+\beta z'\sqrt{s-q})}{\e'\cosh (h\xi+\beta z\sqrt{q}+\beta z'\sqrt{s-q})}\Bigr),&\,\,\mbox{if $s\in [q,1]$},\\
			\\
			\e\bigl(\e' \tanh(h\xi+\beta z\sqrt{s}+\beta z'\sqrt{q-s})\bigr)^2,&\,\,\mbox{if $s\in [0,q)$}
		\end{array}\right.
	\end{align}
	and
	\begin{align}\label{dphi}
		\phi_\alpha'(s)&=\left\{
		\begin{array}{ll}
			\beta^2\e \Bigl(\frac{\e' \cosh^{-3}(h\xi+\beta z \sqrt{q}+\beta z'\sqrt{s-q})}{\e'\cosh (h\xi+\beta z \sqrt{q}+\beta z'\sqrt{s-q})}\Bigr),&\mbox{if $s\in[q,1]$},\\
			\\
			\beta^2\e\bigl(\e '\cosh^{-2}(h\xi+\beta z\sqrt{s}+\beta z'\sqrt{q-s})\bigr)^2,&\mbox{if $s\in [0,q)$}.
		\end{array}
		\right.
	\end{align}
\end{lemma}

\begin{proof}
	A direct computation gives that
	\begin{align*}
		\Phi_{\alpha}(s,x)&=\left\{
		\begin{array}{ll}\frac{\beta^2}{2}(1-s)+\log\cosh x,&\,\,\mbox{if $s\in [q,1]$},\\
			\\
			\frac{\beta^2}{2}(1-q)+\e \log\cosh(x+\beta z\sqrt{q-s}),&\,\,\mbox{if $s\in [0,q)$}
		\end{array}\right.
	\end{align*} 
	and that
	\begin{align*}
		\partial_x\Phi_{\alpha}(s,x)&=
		\left\{
		\begin{array}{ll}\tanh x,&\,\,\mbox{if $s\in [q,1]$},\\
			\\
			\e \tanh(x+\beta z\sqrt{q-s}),&\,\,\mbox{if $s\in [0,q)$}.
		\end{array}\right.
	\end{align*}
	Plugging these along with the assumption $\alpha=1_{[q,1]}$ into Lemma \ref{lem1} establishes \eqref{lem2:eq1}.
	
	As for \eqref{dphi}, note that for any twice differentiable function $f$ with $\|f''\|_\infty<\infty$, we can compute by using the Gaussian integration by parts\footnote{Let $Z=(z_1,\ldots,z_n)$ be a centered Gaussian random vector and let $F$ be a differentiable function on $\mathbb{R}^n$ with $\sum_{i=1}^n\|\partial_{x_i}F\|_\infty<\infty.$ We have that $\e z_1f(Z)=\sum_{i=1}^n \e[z_1z_i]\e\bigl[\partial_{x_i}f(Z)\bigr].$} to obtain
	\begin{align}\label{int}
		\frac{d}{ds}\e' f(z'\sqrt{s-q})&=\frac{1}{2\sqrt{s-q}}\e'f'(z'\sqrt{s-q})=\frac{1}{2}\e'f''(z'\sqrt{s-q}),\,\,\forall s\in (q,1].
	\end{align}
	Applying this equation and 
	\begin{align*}
		\bigl(\tanh^2(x)\cosh(x)\bigr)''&=\frac{2}{\cosh^3(x)}+\tanh^2(x)\cosh(x),\\
		\bigl(\cosh(x)\bigr)''&=\cosh(x)
	\end{align*}
	to the $\e'$-expectations in the first equation of \eqref{lem2:eq1}, the first equation of \eqref{dphi} follows by a straightforward computation.
	To obtain the second equation in \eqref{dphi}, write
	\begin{align*}
		\e\bigl(\e' \tanh(h\xi+\beta z\sqrt{s}+\beta z'\sqrt{q-s})\bigr)^2&=\e\bigl[ \tanh(h\xi+\beta z\sqrt{s}+\beta z_1\sqrt{q-s})\\
		&\qquad\quad\tanh(h\xi+\beta z\sqrt{s}+\beta z_2\sqrt{q-s})\bigr]
	\end{align*}
	for $z_1,z_2$ i.i.d. standard Gaussian independent of $z.$  From the last equation and $(\tanh(x))'={\cosh^{-2}(x)},$ we have
	\begin{align*}
		&\frac{d}{ds}\e\bigl(\e' \tanh(h\xi+\beta z\sqrt{s}+\beta z'\sqrt{q-s})\bigr)^2\\
		&=\frac{\beta}{2}\e \frac{\tanh(h\xi+\beta z\sqrt{s}+\beta z_1\sqrt{q-s})}{\cosh^2(h\xi+\beta z\sqrt{s}+\beta z_2\sqrt{q-s})}\Bigl(\frac{z}{\sqrt{s}}-\frac{z_2}{\sqrt{q-s}}\Bigr)\\
		&+	\frac{\beta}{2}\e \frac{\tanh(h\xi+\beta z\sqrt{s}+\beta z_2\sqrt{q-s})}{\cosh^2(h\xi+\beta z\sqrt{s}+\beta z_1\sqrt{q-s})}\Bigl(\frac{z}{\sqrt{s}}-\frac{z_1}{\sqrt{q-s}}\Bigr).
	\end{align*}
	Using   Gaussian integration by parts yields that
	\begin{align*}
		&\frac{d}{ds}\e\bigl(\e' \tanh(h\xi+\beta z\sqrt{s}+\beta z'\sqrt{q-s})\bigr)^2\\
		&=\beta^2\e \frac{1}{\cosh^2(h\xi+\beta z\sqrt{s}+\beta z_1\sqrt{q-s})\cosh^2(h\xi+\beta z\sqrt{s}+\beta z_2\sqrt{q-s})}\\
		&=\beta^2\e\bigl(\e '\cosh^{-2}(h\xi+\beta z\sqrt{s}+\beta z'\sqrt{q-s})\bigr)^2.
	\end{align*}
\end{proof}

\section{Proof of Theorem \ref{thm1}}\label{sec3}

Throughout the entire proof, we let $q$ be the unique solution to \eqref{fixedpt}. 
First, we assume that $(\beta,h)$ lies above the AT line, i.e.,
\begin{align*}
	\beta^2\e \frac{1}{\cosh^4(h\xi+\beta z\sqrt{q})}> 1.
\end{align*}
We show that $\alpha_P$ can not be replica symmetric. If not, then we must have that $\alpha_P=1_{[q',1]}$ for some $q'\in [0,1]$ and from Proposition \ref{prop2} and Lemma \ref{lem2}, $q'$ must satisfy $$\e\tanh^2(h\xi+\beta z\sqrt{q'})=\phi_{\alpha_P}(q')=q'.$$ 
Since this equation has only one unique solution (see \eqref{fixedpt}), we must have that $q=q'$. Consequently, $\phi_{\alpha_P}(q)=q$ and from \eqref{dphi},  
\begin{align*}
	\phi_{\alpha_P}'(q)&=\beta^2\e \frac{1}{\cosh^4(h\xi+\beta z\sqrt{q})}>1.
\end{align*}
If we let $\varepsilon>0$ be small enough and $\alpha_1=1_{[q+\varepsilon,1]}$, then $\phi_{\alpha_P}(s)>s$ for $s\in [q,q+\varepsilon]$ so that from Proposition \ref{prop1},
$$
\frac{d}{d\theta^+}\mathcal{P}(\alpha_\theta)\Big|_{\theta=0}=-\frac{\beta^2}{2}\int_q^{q+\varepsilon}(\phi_{\alpha_P}(s)-s)ds<0
$$
and consequently, $1_{[q,1]}$ can not be the optimizer. Hence, outside the AT line, the SK model is not replica symmetric.

Next, we assume that $(\beta,h)$ lies inside the AT line, i.e., \eqref{inat} holds. We proceed to show that $\alpha_P=1_{[q,1]}$. To this end, let $\alpha_0=1_{[q,1]}$ and we claim that  \eqref{inat} implies that $\phi_{\alpha_0}(s)\geq s$ if $s<q$ and $\phi_{\alpha_0}(s)\leq s$ if $s>q.$ If this claim is valid, then for any $\alpha_1\in \mbox{Pr}[0,1],$
\begin{align*}
	\frac{d}{d\theta}\mathcal{P}(\alpha_\theta)\Big|_{\theta=0+}&=\frac{\beta^2}{2}\Bigl(\int_0^q\alpha_1(s)(\phi_{\alpha_0}(s)-s)ds+\int_q^1(\alpha_1(s)-1)(\phi_{\alpha_0}(s)-s)ds\Bigr)\geq 0.
\end{align*}
Hence, from Proposition \ref{prop1}, $\alpha_0=1_{[q,1]}$ is the minimizer of the Parisi formula and this will complete our proof. 

We now turn to the proof of our claim. First of all, from \eqref{dphi} and the Jensen inequality with respect to $\e'$, for $s\in [0,q),$ 
\begin{align*}
	\phi_{\alpha_0}'(s)&\leq \beta^2\e\e' \cosh^{-4}(h\xi+\beta z\sqrt{s}+\beta z'\sqrt{q-s})=\beta^2\e \cosh^{-4}(h\xi+\beta z\sqrt{q})\leq 1.
\end{align*}
Since $\phi_{\alpha_0}(q)=q$ by \eqref{lem2:eq1}, we see that $\phi_{\alpha_0}(s)\geq s$ for $s\in [0,q).$ Next, for $s\in (q,1),$ recall from the first equation in \eqref{dphi} that by denoting $Y=h\xi+\beta z\sqrt{q}+\beta z'\sqrt{s-q} $,
\begin{align*}
	\phi_{\alpha_0}'(s)&=\beta^2\e \frac{\e' \cosh^{-3}Y}{\e'\cosh Y}
\end{align*}
Write $\cosh^{-3}x=(\cosh^{-4}x)(\cosh x)$. Since $\cosh^{-4}x$ is decreasing and $\cosh x$ is increasing for $x>0,$ applying the FKG inequality\footnote{The FKG inequality states that if $X$ is a random variable and $f,g$ are both nondecreasing functions with $\mbox{Var}(f(X))<\infty$ and $\mbox{Var}(g(X))<\infty$, then $\e f(X)g(X)\geq \e f(X) \e g(X).$} implies
\begin{align*}
	\e'\cosh^{-3}Y&=\e'\cosh^{-3}|Y|\\
	&\leq \e'\cosh^{-4}|Y|\cdot \e'\cosh|Y|\\
	&=\e'\cosh^{-4}Y\cdot\e'\cosh Y.
\end{align*}
Thus, for any $q<s\leq 1,$
\begin{align*}
	\phi_{\alpha_0}'(s)&\leq  \beta^2\e \e'\cosh^{-4}Y\\
	&=\beta^2\e \cosh^{-4}(h\xi+\beta z \sqrt{s})\\
	&=\beta^2\e \cosh^{-4}(z\sqrt{h^2+\beta^2s})\\
	&\leq \beta^2\e \cosh^{-4}(z\sqrt{h^2+\beta^2q})\\
	&=\beta^2\e \cosh^{-4}(h\xi+\beta z \sqrt{q})\leq 1,
\end{align*}
where the second inequality is valid since $\cosh^{-4}x$ is even and decreasing in $x>0.$ 
Consequently, 
$\phi_{\alpha_0}'(s)-1\leq 0$ for any $q<s\leq 1.$
Since $\phi_{\alpha_0}(q)=q,$ we deduce that $\phi_{\alpha_0}(s)\leq s$ for $q<s\leq 1.$
This finishes the proof of our claim.

\begin{figure}
	\includegraphics[width=5cm]{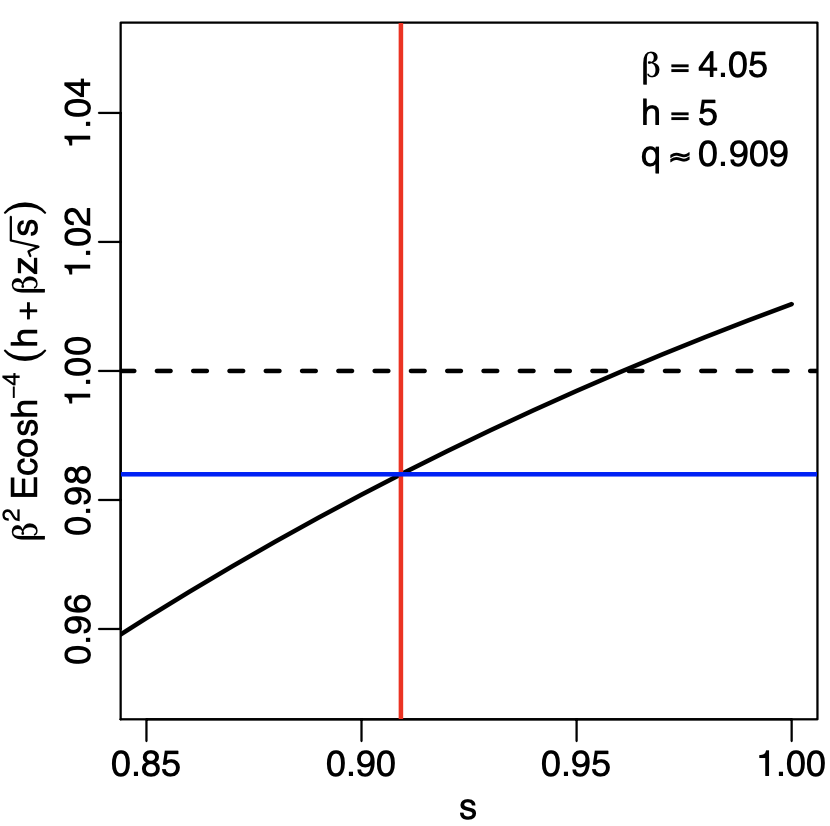}
	\centering
	\caption{This is a simulation of $s\in [q,1]\mapsto \beta^2\e {\cosh^{-4}(h+\beta z\sqrt{s})}$ with $\beta=4.05$, $h=5$, and $q\approx0.909.$ The blue line is the level $ \beta^2\e {\cosh^{-4}(h+\beta z\sqrt{q})})$ and the red line is $q.$ It can be seen that $\beta^2\e {\cosh^{-4}(h+\beta z\sqrt{s})}$ lies above $\beta^2\e {\cosh^{-4}(h+\beta z\sqrt{q})}$ for $s\in [q,1]$ and above $1$ for $s$ large enough.} 
	\label{figure3}
\end{figure}

\begin{remark}
	\rm In the case of the SK model with non-random external field \eqref{SK}, the corresponding $\phi_{\alpha}(s)$ for $\alpha=1_{[q,1]}$ and $q$ satisfying \eqref{fxpt} is the same as in Lemma \ref{lem2} except that $h\xi$ is replaced by $h.$ The same argument in Theorem \ref{thm1} enables us to show that the Parisi formula can not be solved by the replica symmetric solution if $(\beta,h)$ lies outside the AT line. However, if $(\beta,h)$ lies inside the AT line, while it can still be shown that $\phi_\alpha(s)\geq s$ for all $s\in [0,q)$, it is unclear why $\phi_\alpha(s)\leq s$ for all $s\in [q,1].$ In this case, one can still use the FKG inequality to obtain that for $q\leq s\leq 1,$
	\begin{align*}
		\phi_\alpha'(s)&\leq \beta^2\e \cosh^{-4}(h+\beta z\sqrt{s}),
	\end{align*} 
	but a numerical simulation, see Figure \ref{figure3}, suggests that the following does not always hold:
	$$
	\beta^2\e \cosh^{-4}(h+\beta z\sqrt{s})\leq  \beta^2\e \cosh^{-4}(h+\beta z\sqrt{q})
	$$
	for all $q\leq s\leq 1.$
\end{remark}

\bibliographystyle{amsplain}
\footnotesize{\bibliography{ref}}

\providecommand{\bysame}{\leavevmode\hbox to3em{\hrulefill}\thinspace}
\providecommand{\MR}{\relax\ifhmode\unskip\space\fi MR }
\providecommand{\MRhref}[2]{%
  \href{http://www.ams.org/mathscinet-getitem?mr=#1}{#2}
}
\providecommand{\href}[2]{#2}
\begin{thebibliography}{10}

\bibitem{AT78}
J.~F.~L. Almeida and D.~J. Thouless, \emph{Stability of the
  {S}herrington-{K}irkpatrick solution of a spin glass model}, J. Phus. A:
  Math. Gen. \textbf{II} (1978), 983--990.

\bibitem{AC15}
A.~Auffinger and W.-K. Chen, \emph{The {P}arisi formula has a unique
  minimizer}, Comm. Math. Phys. \textbf{335} (2015), no.~3, 1429--1444.
  \MR{3320318}

\bibitem{Chen17}
W.-K. Chen, \emph{Variational representations for the {P}arisi functional and
  the two-dimensional {G}uerra-{T}alagrand bound}, Ann. Probab. \textbf{45}
  (2017), no.~6A, 3929--3966. \MR{3729619}

\bibitem{Guerra01}
F.~Guerra, \emph{Sum rules for the free energy in the mean field spin glass
  model}, Fields Institute Communications \textbf{30} (2001), no.~161.

\bibitem{Guerra03}
\bysame, \emph{Broken replica symmetry bounds in the mean field spin glass
  model}, Comm. Math. Phys. \textbf{233} (2003), no.~1, 1--12. \MR{1957729}

\bibitem{JT16}
A.~Jagannath and I.~Tobasco, \emph{A dynamic programming approach to the
  {P}arisi functional}, Proc. Amer. Math. Soc. \textbf{144} (2016), no.~7,
  3135--3150. \MR{3487243}

\bibitem{JT17}
\bysame, \emph{Some properties of the phase diagram for mixed {$p$}-spin
  glasses}, Probab. Theory Related Fields \textbf{167} (2017), no.~3-4,
  615--672. \MR{3627426}

\bibitem{MPV87}
M.~M\'{e}zard, G.~Parisi, and M.~A. Virasoro, \emph{Spin glass theory and
  beyond}, World Scientific Lecture Notes in Physics, vol.~9, World Scientific
  Publishing Co., Inc., Teaneck, NJ, 1987. \MR{1026102}

\bibitem{Pan13}
D.~Panchenko, \emph{The {S}herrington-{K}irkpatrick model}, Springer Monographs
  in Mathematics, Springer, New York, 2013. \MR{3052333}

\bibitem{Parisi79}
G.~Parisi, \emph{Infinite number of order parameters for spin-glasses}, Phys.
  Rev. Lett. \textbf{43} (1979), 1754--1756.

\bibitem{Parisi80}
\bysame, \emph{A sequence of approximate solutions to the {SK} model for spin
  glasses}, J. Phys. A. \textbf{13} (1980), no.~L-115.

\bibitem{SK72}
D.~Sherrington and S.~Kirkpatrick, \emph{Solvable model of a spin glass}, Phys.
  Rev. Lett. \textbf{35} (1972), 1792--1796.

\bibitem{Tal06}
M.~Talagrand, \emph{The {P}arisi formula}, Ann. of Math. (2) \textbf{163}
  (2006), no.~1, 221--263. \MR{2195134}

\bibitem{Tal111}
\bysame, \emph{Mean field models for spin glasses. {V}olume {I}}, Ergebnisse
  der Mathematik und ihrer Grenzgebiete. 3. Folge. A Series of Modern Surveys
  in Mathematics [Results in Mathematics and Related Areas. 3rd Series. A
  Series of Modern Surveys in Mathematics], vol.~54, Springer-Verlag, Berlin,
  2011, Basic examples. \MR{2731561}

\bibitem{Tal112}
\bysame, \emph{Mean field models for spin glasses. {V}olume {II}}, Ergebnisse
  der Mathematik und ihrer Grenzgebiete. 3. Folge. A Series of Modern Surveys
  in Mathematics [Results in Mathematics and Related Areas. 3rd Series. A
  Series of Modern Surveys in Mathematics], vol.~55, Springer, Heidelberg,
  2011, Advanced replica-symmetry and low temperature. \MR{3024566}

\bibitem{Toninelli02}
F.~Toninelli, \emph{About the {A}lmeida-{T}houless transition line in the
  {S}herrington-{K}irkpatrick mean-field spin glass model}, Europhys. Lett.
  \textbf{60} (2020), 764--767.

\end{thebibliography}
\end{document}